\documentclass[1p]{elsarticle}

\usepackage{ifpdf}
\usepackage{graphicx,amssymb,lineno}
\usepackage{epsfig}
\usepackage{graphics}
\usepackage[emtex]{color}
\usepackage[usenames,dvipsnames,svgnames,table]{xcolor}
\usepackage{amsfonts}
\usepackage {eepic}
\usepackage{latexsym}
\usepackage{amsmath,amsthm}
\usepackage{graphics}
\usepackage{dsfont}
\usepackage{euscript}
\usepackage{mathrsfs}
\usepackage{graphics}
\usepackage {epic}
\usepackage {eepic}
\usepackage{latexsym}
\usepackage{fancybox}
\usepackage[active]{srcltx}
\usepackage{mathpazo}
\usepackage[lined,boxruled,algosection,algo2e,linesnumbered]{algorithm2e}
\usepackage{setspace}

\usepackage{soul}

\usepackage{enumerate}

\newtheorem{theorem}{Theorem}[section]
\newtheorem{lemma}[theorem]{Lemma}

\newtheorem{corollary}[theorem]{Corollary}

\newcommand{\NP}{{\mathsf{NP}}}
\newcommand{\PP}{{\mathsf{P}}}

\newcommand{\ZPP}{{\mathsf{ZPP}}}

\newcommand{\maxkcomp}[1]{\textsc{Max-}{$#1$}\textsc{-OBCS}}
\newcommand{\maxwkcomp}[1]{\textsc{Max-Weighted-}$#1$\textsc{-OBCS}}

\newcommand{\disso}{{\textsc{Max-Disso-Set}}}

\usepackage{todonotes}

\begin{document}
\begin{frontmatter}

\title{The Complexity of Maximum $k$-Order Bounded Component Set Problem}
\author{Sounaka Mishra \fnref{SM}}
\address{Department of Mathematics, Indian Institute of Technology Madras,
Chennai 600 036, India}
\ead{sounak@iitm.ac.in}
\fntext[SM]{Corresponding author}
\author{Shijin Rajakrishnan }
\address{Department of Computer Science, Cornell University, New York, USA}
\ead{sr986@cornell.edu}

\begin{abstract} 
Given a graph $G=(V, E)$ and a positive integer $k$, in \textsc{Maximum $k$-Order Bounded Component Set} (\maxkcomp{k}), it is required to find a vertex set $S \subseteq V$ of maximum size such that each component in the induced graph $G[S]$ has at most $k$ vertices. We prove that for constant $k$, \maxkcomp{k} is hard to approximate within a factor of $n^{1 -\epsilon}$, for any $\epsilon > 0$, unless $\PP = \NP$. This is an improvement on the previous lower bound of $\sqrt{n}$ for \maxkcomp{2} due to Orlovich et. el. \cite{orlovich2011complexity}. We provide lower bounds on the approximability when $k$ is not a constant as well. \maxkcomp{k} can be seen as a generalization of \textsc{Maximum Independent Set} (\textsc{Max-IS}). We generalize Tur\'{a}n's greedy algorithm for \textsc{Max-IS} and prove that it approximates \maxkcomp{k} within a factor of $(2k - 1)\overline{d} + k$, where $\overline{d}$ is the average degree of the input graph $G$. This approximation factor is a generalization of Tur\'{a}n's approximation factor for \textsc{Max-IS}.
\end{abstract}

\end{frontmatter}

\section{Introduction}
In this paper, we consider the computational complexity of an optimization problem called the \textsc{Maximum $k$-Order Bounded Component Set} problem  (\maxkcomp{k}).  In \maxkcomp{k}, given a graph $G=(V, E)$ and a positive integer $k$, it is required to find a vertex set $S \subseteq V$ of maximum size such that order of each component in $G[S]$ is at most $k$.

This is a generalization of several important problems in combinatorics. For $k=1$, this problem is the same as the \textsc{Maximum Independent Set} problem (\textsc{Max-IS}), where given a graph $G=(V,E)$, the objective is to find a subset of vertices $S \subseteq V$ of maximum size such that there is no edge between any pair of vertices in $G[S]$. The case of $k=2$ gives rise to the \textsc{Maximum Dissociation Set} problem (\disso), where given a graph $G$, the objective is to find a vertex set $S \subseteq V$ of maximum size such that the degree of each vertex in $G[S]$ is at most $1$. Both these problems are {\sf NP}-hard and thus \maxkcomp{k} is {\sf NP}-hard as well. 

Finding a maximum sized independent set is one of the widely studied problems in combinatorics and its hardness to compute has prompted several attempts at developing approximate solutions.  
H\aa{}stad \cite{hastad1996clique} proved that , for any $\epsilon>0$,  \textsc{Max-IS} cannot be approximated within a factor of $n^{1-\epsilon}$, unless $\NP = \ZPP$. Later, Zukerman \cite{zuckerman2006linear} proved that  \textsc{Max-IS} cannot be approximated within a factor of $n^{1-\epsilon}$, unless $\PP = \NP$, for any $\epsilon > 0$. This lower bound result rules out the existence of a polynomial time efficient algorithm on general graphs. However, the situation is not quite as grim when restricted to special graph classes. It is polynomial time solvable to find a maximum independent set in claw-free graphs \cite{nakamura2001revision}, $P_5$-free graphs \cite{lokshantov2014independent}, and perfect graphs \cite{grotschel2012geometric}. For other classes, although finding a polynomial time optimal solution might be out of the question unless $\PP=\NP$, we can nevertheless provide good approximation algorithms. In planar graphs, there exists a PTAS for \textsc{Max-IS} \cite{baker1994approximation} whereas for graphs with bounded degree, the problem can be approximated within a constant factor \cite{papaYa}. One such algorithm greedily includes a vertex of minimum degree into the solution set, deletes that vertex along with its neighbors from the graph and then repeats this selection process until there are no more vertices to process.  This simple greedy algorithm returns an independent set which is within a factor of $(\overline{d} + 1)$ from the optimal solution (Tur\'{a}n's Theorem)~\cite{hochbaum1983efficient}. A tighter analysis has improved the approximation factor to $\frac 15(2\overline{d} + 3)$~\cite{halldorsson1997greed}.

In \cite{YannakakisM}, Yannakakis defined dissociation set and proved that it is $\NP$-complete to find a vertex set $S$ of minimum size in a given bipartite graph $G$ such that $G[V \setminus S]$ has maximum degree 1. From this result it follows that \disso{} is $\NP$-complete even when restricted to bipartite graphs. It is also known to be $\NP$-complete for $K_{1,4}$-free bipartite graphs, $C_4$-free bipartite graphs of maximum degree 3~\cite{Boliac2004OnCT}, planar graphs with maximum degree 4~\cite{papadimitriou1982complexity} and planar line graphs of planar bipartite graphs~\cite{orlovich2011complexity}.  
However, \disso{} is polynomial time solvable when the input graph is restricted to some hereditary classes of graphs, such as   $\{chair, bull\}$-free graphs, $\{chair, K_3\}$-free graphs, $\{mK_2\}$-free graphs with $m \geq 2$ \cite{orlovich2011complexity}.

Approximation algorithms for the complementary problem of \maxkcomp{k} have been designed by various researchers. In this minimization problem, we are asked to find a vertex set $S$ of minimum size in a given graph $G$ such that each component in $G[V \setminus S]$ has at most $k$ vertices. When $k=1$ and $2$, these minimization problems are known as \textsc{Minimum Vertex Cover} and  \textsc{Minimum Dissociation Set}, respectively.
In \textsc{Minimum Dissociation Set} it is required to find a vertex set $S$ of minimum cardinality in a given graph $G=(V, E)$ such that degree of each vertex in $G[V\setminus S]$ is at most 1. It is interesting to observe that \textsc{Minimum Dissociation Set} is equivalent to \textsc{Minimum $P_3$ Vertex Cover} in which it is required to find a vertex set $S$ of minimum size such that $G[V \setminus S]$ does not have a $P_3$ (path on 3 vertices) as a subgraph.  Both \textsc{Minimum Vertex Cover} and \textsc{Minimum Dissociation Set} are known to be approximable within  a factor of $2$ \cite{fujito}.   For larger values of $k$, the deletion problem is known to be approximable within a factor of $k$ and approximable within a factor of $(k-1)$ when the input graph has girth at least $k$~\cite{zhang2014approximation}. However, to the best of our knowledge, there is no known result about the approximability of \maxkcomp{k} for arbitrary graphs.\\

\noindent
{\bf Our Contribution.} In this paper, we explore the approximability and inapproximability of  \maxkcomp{k}, for different integral values of $k$. In Section 3, we prove inapproximability results by establishing an approximation preserving reduction from \textsc{Max-IS}, which shows that for constant $k$, this problem has no efficient approximation algorithms for general graphs. The reduction used can be composed to prove hardness results even for non-constant $k$. In particular, we show that for $k= O(\frac{n}{\log n})$, it is not possible to develop an algorithm with a better approximation guarantee than $O(\log n)$, unless $\PP = \NP$. Having established the hardness of computing efficient solutions to this problem, in the subsequent sections we concentrate on developing upper bounds. In Section 4, we prove that \maxwkcomp{k} can be approximated within a factor of $\Delta$, where $\Delta$ is the maximum degree of the input graph $G$. This also proves that the problem can be approximated within a constant factor for degree bounded graphs. It is important to note that this bound holds not just for the unweighted case but also for the weighted version of the problem, where there is a weight function $w: V \to \mathbb{R}^+$ as an additional input and the objective is to find a vertex set $S$ of maximum weight such that $G[S]$ has no component whose order exceeds $k$. The technique used in the algorithm is the Local Ratio Method~\cite{bar2004local}. In Section 5, we build upon the algorithm for \textsc{Max-IS} to develop a greedy algorithm for \disso, with a performance guarantee of $(3\overline{d} + 2)$, where $\overline{d}$ is the average degree of the graph. We then generalize the algorithm and the analysis to the larger problem with arbitrary values of $k$ and prove that the extension of the algorithm yields a solution which is within a factor of $(2k-1)\overline{d} + k$ from the optimal solution. This can be observed as a generalization of Tur\'{a}n's bound mentioned earlier for \textsc{Max-IS}. Finally, in Section 6, we establish a reduction to \textsc{Max-IS} and use existing upper bounds for \textsc{Max-IS}~\cite{halldorsson2004approximations} to prove that \maxkcomp{k} can be approximated within a factor of $O(\frac{k\Delta \log \log \Delta}{\log \Delta})$, for sufficiently large values of $\Delta$.

\section{Notations}
Throughout this paper, we assume that any graph $G=(V, E)$ mentioned is simple and undirected, with $|V|=n$ vertices and $|E|=m$ edges. Given a graph $G = (V,E)$, the open neighborhood $N_G(v)$ of a vertex $v$ is the set of vertices those are adjacent to $v$ in $G$ and the closed neighborhood $N_G[v]$ of $v$ is defined as $N_G[v] = N_G(v)\cup \{v\}$. For a given set $P \subseteq V$, we define $N_{P}(v) = N_G(v) \cap P$, for each $v \in V$.
The degree of a vertex $v$ in the graph $G$ is defined as $d_G(v)= |N_G(v)|$, the maximum degree of the graph is $\Delta_G = \max_{v \in V}d_G(v)$. When the underlying graph $G$ is unambiguous, we drop the subscript, for example we use just $\Delta$ to denote the maximum degree in $G$ instead of $\Delta_G$. The average degree of the graph $G$ is $\overline{d}=\frac 1n \sum_{v \in V}d_G(v)$. There is an interesting relationship between the average degree of a graph and the total number of edges in the graph, which we use in this paper, viz. the sum of the degrees of all the vertices in a graph is twice the number of edges. So $\sum_{v\in V}d_G(v) = 2|E|$, which gives us $|V|\overline{d} = 2|E|$.

Given a subset of vertices, $S \subseteq V$, the graph induced on this set, $G[S]$, is the subgraph whose vertex set is $S$ and the edge set is the subset of edges $E$ such that both endpoints are in $S$, i.e, $G[S] = (V',E')$ where $V' = S$ and $E' = \{(u,v) \in E \mid u \in S \mbox{ and } v \in S \}$.

A vertex set $S \subseteq V$ is called an independent set if $G[S]$ has no edges. We shall denote $\alpha(G)$ as the size of a maximum independent set in $G$. An independent set $S$ is {\em maximal} if it is not a proper subset of another independent set of $G$. A set of vertices $S$ of $G$ is called a dissociation set if degree of each vertex in $G[S]$ is at most 1 (equivalently, each component in $G[S]$ has at most two verices). In the maximum dissociation set problem (\disso), the objective is to find a dissociation set of maximum size in a given graph $G$. 
A set $S \subseteq V$ is called a $k$-component set in $G$  if each component in $G[S]$ has at most $k$ vertices. We shall denote the size of a maximum $k$-component set in $G$ by ${\sf comp}_k(G)$. 

\section{Hardness of approximation}
In this section, we shall prove inapproximability results for \maxkcomp{k} by establishing an {\sf AP}-reduction~\cite{ACP95} from \textsc{Max-IS}.
The following lower bound result for \textsc{Max-IS} makes this reduction useful.

\begin{theorem} {\rm~\cite{zuckerman2006linear}} \label{thmHastad}
Unless $\PP = \NP$, for any $\epsilon > 0$, \textsc{Max-IS} is not  approximable within a factor of $n^{1 - \epsilon}$.
\end{theorem}

The hardness result mentioned in Theorem \ref{thmHastad} establishes the hardness of approximation for \maxkcomp{1}. Next, we establish the hardness of approximation for \maxkcomp{k} for any constant positive integer $k$. We make use of the following reduction for this purpose. 

\begin{lemma} \label{lem:k-2k-reduction}
 For any $k>0$, \maxkcomp{k} $\leq_{\text{\sf AP}}$ \maxkcomp{2k}.   
\end{lemma}
\begin{proof}
Given an instance $G = (V,E)$ of \maxkcomp{k}, we construct an instance $G' = (V',E')$ of \maxkcomp{2k} as follows. First we make two copies $G_1=(V_1, E_1)$ and $G_2=(V_2, E_2)$ of the graph $G$ with $V_1= \{v_1 \mid v \in V\}$, $V_2 = \{v_2 \mid v \in V\}$ and $E_1, E_2$ are defined accordingly.
Then, for each edge $(u, v) \in E$, we add two edges $(u_1, v_2)$ and $(u_2, v_1).$  Finally, we introduce the edge $(u_1, u_2)$ for each $u \in V$. Thus we have $V' = V_1 \cup V_2$ and $E' = E_1 \cup E_2 \cup \{(u_1,v_2), (u_2, v_1) \mid (u,v) \in E\} \cup \{(u_1, u_2) \mid u \in V\}$. It is easy to observe that $|V'| = 2|V|$. For an illustration we refer to Figure \ref{lemma3.2}.

\begin{figure}[h]
\begin{center}
    \includegraphics[scale=.5]{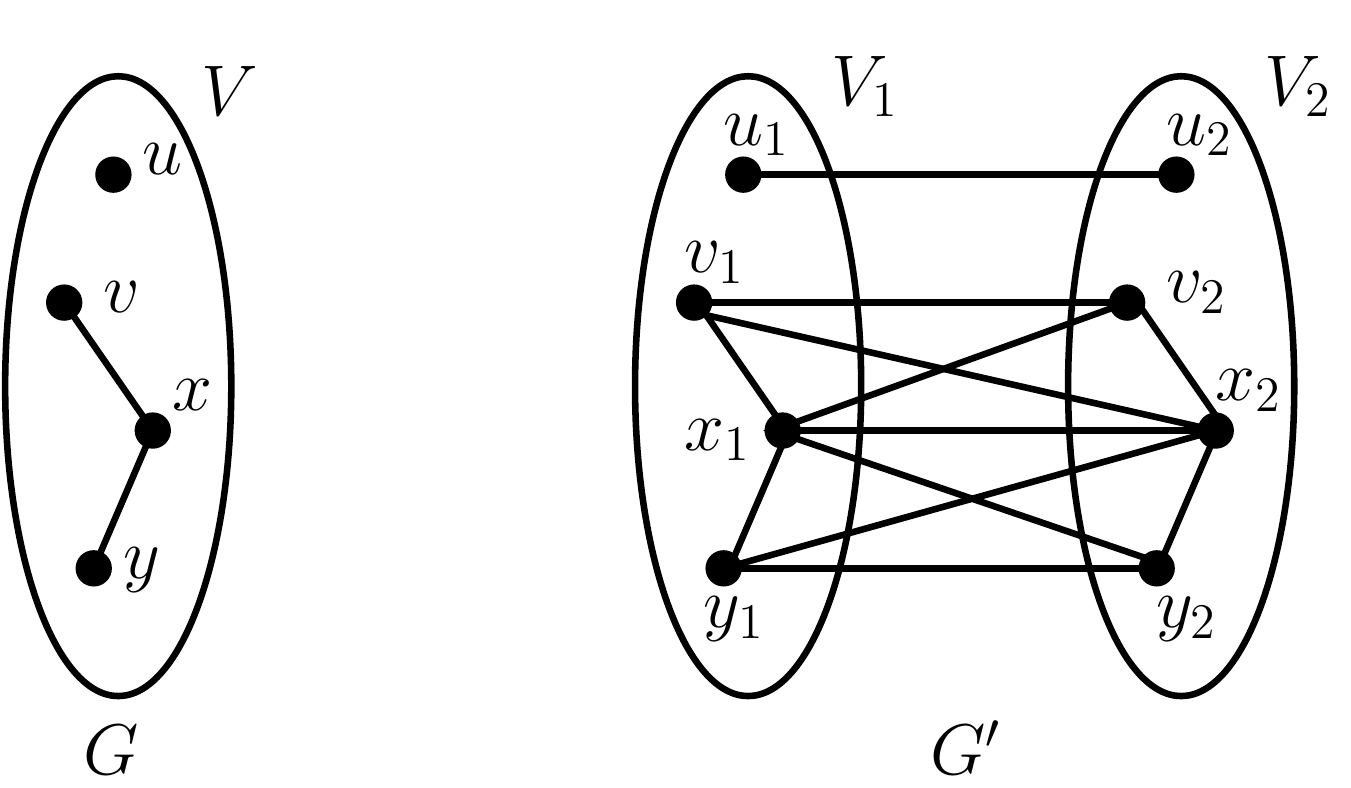}
\end{center}
\caption{Sketch of the construction of $G'$ from $G$}
\label{lemma3.2}
\end{figure}

Let $S^*$ be a maximum $k$-component set in $G$ and $S'^*$ be a maximum $2k$-component set in $G'$. 
From the construction of $G'$ from $G$, it follows that the set $S'= \{v_1, v_2 \mid v \in S\}$ is a $2k$-component set in $G'$, for every $k$-component set $S$ in $G$. Also $|S'|= 2|S|$. Since these two problems are maximization problems we have $2|S^*| \leq |S'^*|.$

Let $S'$ be a $2k$-component set in $G'$, consisting of $t$ components. From $S'$ we will construct a $k$-component set $S$ in $G$ with $|S| \geq \frac{1}{2}|S'|$.

First, we prove a property of the components in $G'[S']$ which is as follows. For any two distinct components $T_1$ and $T_2$, if we look at the vertices in $G$ induced by the components then there are no common vertices. In other words, the sets $\{v\in V \mid v_1\in T_1 \mbox{ or } v_2 \in T_1\}$ and $\{v\in V \mid v_1 \in T_2 \mbox{ or } v_2 \in T_2\}$ are disjoint. Suppose that these two sets  have a common vertex, say $u$. Then by construction,  $u_1 \in T_1$ and $u_2 \in T_2$ and $(u_1, u_2) \in E'$. This implies that $T_1$ and $T_2$ are connected, which is a contradiction.

Now, we construct a set $S \subseteq V$ from $S'$ as follows. Let $T_1, T_2, \ldots, T_t$ be the components in $G'[S']$. For each component $T_i$ in $G'[S']$, if $|T_i\cap V_1| \geq |T_i \cap V_2|$ and $|T_i \cap V_1| \geq k$ choose a set $J_i$ of any $k$ vertices from $T_i \cap V_1$; otherwise if $|T_i \cap V_1| < k$, take $J_i = T_i \cap V_1.$ Finally, we update $S=S \cup \{u \in V \mid u_1 \in J_i\}$. If $|T_i \cap V_2| > |T_i \cap V_1|$, then we take vertices from $V_2$ into the set $J_i$ and update $S = S \cup \{u\in V \mid u_2\in J_i\}$.  

Let $S= J_1 \cup J_2 \cup \ldots \cup  J_t$, where $J_i$ is the vertex set selected from the component $T_i$ of $G'[S']$. We claim that $S$ is a $k$-component set in $G$. It is easy to observe that $G[J_i]$  can have at most $k$ vertices. Next, we need to prove that there is no edge between a vertex $u$ in $J_i$ and a vertex $v$ in $J_j$, for $i \neq j$. The existence of such an edge would imply that the corresponding components $T_i$ and $T_j$ are not distinct (because of the edges $(u_1, v_1), (u_1, v_2), (u_2, v_1), (u_2, v_2)$). 

From the construction of vertex set $J_i$ from the component $T_i$, it follows that cardinality of $J_i$ is at least half of the vertices in $T_i$. Therefore, $|S| \geq \frac{1}{2}|S'|.$

Therefore, for any given $2k$-component set $S'$ in $G'$, one can construct a $k$-component set $S$ in $G$ in polynomial time, such that $\frac{|S^*|}{|S|} \leq \frac{|S'^*|}{|S'|}.$ 
\end{proof}

\begin{corollary} \label{cor3.3}
 For $k>0$, \textsc{Max-IS} is {\sf AP}-reducible to \maxkcomp{k} with size amplification of $kn$, where $n$ is the number of vertices in the input instance of \textsc{Max-IS} and $k = 2^i$, for some nonnegative integer $i$.
\end{corollary}

\begin{proof}
We reduce \textsc{Max-IS} to \maxkcomp{k} by composing the reduction in Lemma \ref{lem:k-2k-reduction} $\lceil \log k \rceil$ times to get an instance of \maxkcomp{k}. Since the size of the new instance in the reduction of  Lemma \ref{lem:k-2k-reduction} doubles, on applying the reduction $\lceil \log k \rceil$ times, the size of the \maxkcomp{k} instance is $2^{\lceil \log k \rceil}n \approx kn$. Since the reduction in Lemma \ref{lem:k-2k-reduction} is a ratio preserving reduction, the composite reduction from \textsc{Max-IS} to \maxkcomp{k} is a ratio preserving reduction.  Therefore, \textsc{Max-IS} $\leq_{\text{\sf AP}}$ \maxkcomp{k}.
\end{proof}

Using Theorem \ref{thmHastad} and Corollary \ref{cor3.3} we have the following result for \maxkcomp{k}.

\begin{corollary} \label{clbBcis}
For any $\epsilon >0$, \maxkcomp{2^i} is hard to approximate within a factor of $n^{1 - \epsilon}$,  unless $\PP = \NP$, where $i$ is a constant nonnegative integer.
\end{corollary}

Next we will extend this lower bound result for \maxkcomp{k}, where $k$ is a constant  nonnegative integer.

\begin{theorem} \label{lbBcis}
For any $\epsilon >0$, \maxkcomp{k} is hard to approximate within a factor of $n^{1 - \epsilon}$,  unless $\PP = \NP$, where $k$ is a constant nonnegative integer.
\end{theorem}
\begin{proof}
Now consider a positive integer $k$ such that $2^i < k < 2^{i+1}$ for some non-negative integer $i$. Suppose that \maxkcomp{k} can be approximated within a factor of $\beta$, to obtain a subset $S_k$ such that $|S_k| \geq \frac1\beta|S^*_k|$. We can use this subset $S_k$ to produce a solution for \maxkcomp{2^i}, by considering each component in $S_k$ and dropping vertices arbitrarily from each component until its size is at most $2^i$. Notice that the number of components in $S_k$ with size more than $2^i$ can be at most $\frac{|S_k|}{2^i}$, and the number of vertices we drop from each such component is at most $k - 2^i$. Thus the size of the constructed solution to \maxkcomp{2^i}, $S_{2^i}$ is at least $|S_k| - \frac{|S_k|}{2^i}(k - 2^i) = |S_k|\frac{2^{i+1}-k}{2^i}$. We also know that $|S^*_{2^i}| \leq |S^*_k|$, and thus the approximation ratio of this constructed set is 
\[\frac{|S_{2^i}|}{|S^*_{2^i}|} \geq |S_k|\frac{2^{i+1} -k}{2^i}\frac1{|S^*_k|} \geq \frac{2^{i+1} - k}{2^i}\frac1\beta\]
Now the hardness result of \maxkcomp{2^i} implies that $\beta \in \omega(n^{1-\varepsilon})$, and thus \maxkcomp{k} cannot be approximated within a factor of $n^{1-\varepsilon}$ for any constant $k$, unless {\sf P} = {\sf NP}.
\end{proof}

The authors in \cite{orlovich2011complexity} have proved that, for any $\epsilon > 0$,  \disso{} is hard to approximate within a factor of $n^{\frac12 - \epsilon}$,  unless $\PP = \NP$. When $i=1$, \maxkcomp{2^i} is same as \disso{} and we have the following result which is a stronger negative result for \disso{}.

\begin{theorem} \label{disso}
For any $\epsilon >0$, \disso{} is hard to approximate within a factor of $n^{1 - \epsilon}$,  unless $\PP = \NP$. 
\end{theorem}

Next, we consider the inapproximability of \maxkcomp{k} when $k$ is not a constant and our result is as follows.

\begin{theorem} \label{lower_BcisK}
If $k$ is not a constant, then for any $\epsilon > 0$, \maxkcomp{k} cannot be approximated in polynomial time with a better factor than $n^{\frac{\log (n/k)}{\log (n)} - \epsilon}$,  unless $\PP = \NP$.
\end{theorem}
\begin{proof}
Note that in the previous theorem, the graph $G'$ constructed in the reduction of \maxkcomp{k} from \textsc{Max-IS} has size $n'= n2^{\log k} = nk$, where $n$ is the number of vertices in an instance $G$ of \textsc{Max-IS}.

Now suppose that \maxkcomp{k} can be approximated within a factor of ${n'}^{\frac{\log (n'/k)}{\log (n')} -  \epsilon}$ in polynomial time, then 
\begin{align*}
\frac{\alpha(G)}{|S|} \leq \frac{comp_k(G')}{|S'|} &\leq 
(nk)^{\frac{\log (n)}{\log (nk)} - \epsilon} \\
&= \frac{n}{(nk)^{\epsilon}}\cdot n^{-\frac{\log k}{\log nk}}k^{\frac{\log n}{\log nk}}  \\
&= \frac{n}{(nk)^{\epsilon}}\cdot \left({n^{{-\log k}}n^{{\log k}}}\right)^{\frac{1}{\log nk}} \\
&=\frac{n}{(n)^{\epsilon'}} = n^{1-\epsilon'} < n^{1-\epsilon},
\end{align*}
where $\epsilon' = \epsilon\left[1 + \frac{\log k}{\log n}\right]$.  The last inequality holds as $\epsilon' > \epsilon$. This  contradicts the lower bound for \textsc{Max-IS} available in Theorem \ref{thmHastad}. 
\end{proof}

From Theorem \ref{lower_BcisK}, we have the following corollary,

\begin{corollary}
 If $k = O(\frac{n}{\log n})$, then \maxkcomp{k} is hard to approximate within a factor of $O(\log n)$,  unless $\PP = \NP$.
\end{corollary}

\section{Approximation algorithms for \maxwkcomp{k}}

We now turn to the weighted version of the problem, \maxwkcomp{k}. The input, in addition to a graph $G=(V,E)$, consists of a weight function on the vertices of the graph, $w:V\to \mathbb{R}^+$. Given a subset of vertices $A \subseteq V$, the weight of the subset is defined as $w(A) = \sum_{v \in A} w(v)$. In \maxwkcomp{k}, the objective is to find a subset of vertices $S \subseteq V$ of maximum weight $w(S)$ such that the size of the largest component in $G[S]$ is at most $k$. Since the unweighted case looked at thus far is a special case of the weighted case (with the weight function being uniform), the lower bounds derived in the previous section hold for the weighted version as well. 

In this section, we look at an approximation algorithm for this problem and prove an approximation guarantee of $\Delta$, the maximum degree of the input graph. 

\begin{algorithm2e}[h] 
\SetAlgoLined
\KwIn{An undirected simple graph $G=(V, E)$ and $w:V \rightarrow \mathbb{R}^+$}
\KwOut{$S \subseteq V$ such that each component in $G[S]$ has at most $k$ vertices}
  \If{$G = \emptyset$}{Return $\emptyset$\; \nllabel{basecase}} 
  \eIf{ there exists a vertex $u$ such that $w(u) \leq 0$ }{
   Return $k$-Weighted-OBCS($G \setminus \{u\}$,$w$)\;
   }{
       Let $(s,t)$ be the edge in $G$ that maximizes $w(s) + w(t)$\;
   Define the weight functions \quad
   $w_1(u) = \left\{\begin{array}{ll} w(u), & u\in N[\{s,t\}]\\ 0, & o.w.\end{array}\right.$, and $w_2 = w - w_1$\;
   $S' \gets$ $k$-Weighted-OBCS$(G, w_2)$\; \nllabel{inductive}
   Return $S' \cup \{s,t\}$\;
   }
 \caption{$k$-Weighted-OBCS($G$,$w$) }\label{algo4.1}
\end{algorithm2e}

The algorithm is described in Algorithm \ref{algo4.1} and the analysis makes use of the local ratio theorem.
\begin{theorem} {\rm \cite{bar2004local}}
 Let $(G=(V, E), w)$ be an instance of \maxwkcomp{k} having $n$ vertices. Let $w_1, w_2 \in \mathbb{R}^n$ such that $w = w_1 + w_2$. Let $S \subseteq V$ be an $r$-approximate feasible solution of \maxwkcomp{k} for $G$ with respect to $w_1$ and $w_2$. Then $S$ is an $r$-approximate solution of \maxwkcomp{k} for  $G$ with respect to $w$.
\end{theorem}

By using this theorem we have the following result.
\begin{theorem}
 \maxwkcomp{k} can be approximated within a factor of $\Delta$, where $\Delta$ is the maximum degree of the input graph $G$.
\end{theorem}
\begin{proof}
We prove the correctness by induction. In the base case, Line \ref{basecase} is executed, and it is $\Delta$-approximate. In the inductive step, when Line \ref{inductive} is executed, $S'$ is $\Delta$-approximate with respect to $w_2$.
Now, the only vertices with a non zero weight in $w_1$ are the vertices $s$, $t$, and their neighbors. We will now prove that the weight $w_1(s) + w_1(t)$ is at least a $1/\Delta$ fraction of the total weight of $N[\{s,t\}]$ (with respect to $w_1$), which will prove the theorem. First note that since $(s,t)$ was the edge with the maximum combined weight, we have that for any neighbour $u$ of $s$, $w_1(u) = w(u) \leq w(t) = w_1(t)$, and likewise for any neighbour $v$ of $t$, we have $w_1(v) = w(v) \leq w(s) = w_1(s)$. Now we have,
\[\frac{w_1(\{s,t\})}{w_1(N[\{s,t\}])} \geq \frac{w_1(s) + w_1(t)}{w_1(s)d_G(s) + w_1(t)d_G(t)}\geq \frac1\Delta\]
where the first inequality follows from our observation above, and the second inequality holds since the degree of each vertex is at most $\Delta$.

Thus using the local ratio theorem, we get a $\Delta$-factor approximation algorithm for \maxwkcomp{k}.
\end{proof}

It should be noted that due to the relation between the weighted and unweighted cases, the upper bound derived here holds even for the unweighted case.

\section{Approximation algorithms for \disso{}}
A subset of vertices $S \subseteq V$ in a graph $G=(V, E)$ is called a dissociation set if it induces a subgraph with vertex degree of at most 1. In Theorem \ref{disso} we show that \disso{} is hard to approximate within a factor of $n^{1 - \epsilon}$, unless $\PP = \NP$. In this section, we prove that a greedy algorithm approximates \disso{} within a factor of $3 \overline{d} + 2$.

In this algorithm, we keep track of the components of size 1 and 2 in $G[S]$. $S_1$ is the set of vertices in $S$ which form components of size 1 and $S_2$ is the set of vertices which form components of size 2 in $G[S]$. Each time we include a vertex into $S$, we update the set $X$ such that each component in $G[S\cup X]$ has at least 3 
vertices.

We first prove the correctness of the Algorithm \ref{Disscociation-Algo} and then turn to its approximation guarantee.
Assume that the algorithm runs for $q$ iterations. Since we include one vertex in each iteration, it is clear that the size of the solution set is $q$. Let $v_i$ be the vertex included in the $i$th iteration and so the solution set returned is $S = \{v_1,\ldots,v_q\}$.

\begin{lemma} \label{lemma:S2inX}
At any iteration of the algorithm, for each vertex $v \in S_2$, $N_{V'}(v) \subseteq X$.
\end{lemma}
\begin{proof}
We prove this statement by induction on the iteration number. For the base case $S_2 = \emptyset$ and the statement trivially holds. Assume that the statement is true for the $i$th iteration. Let $v_{i+1}$ be the vertex added to $S$ in $(i+1)$th iteration. If $v_{i+1}$ has no neighbor in $S_1$ then $S_2$ remains unchanged and the statement holds. Otherwise, a new pair of vertices is added to $S_2$ and line \ref{X} ensures that the neighbors of these two vertices are included in $X$. Thus the lemma holds in this case as well.
\end{proof}

\begin{algorithm2e}[H] 
\SetAlgoLined
\KwIn{An undirected simple graph $G=(V, E)$}
\KwOut{$S \subseteq V$ such that each component in $G[S]$ has at most 2 vertices}
 $S_1 = \emptyset$; $S_2= \emptyset$; $X = \emptyset$; $S= \emptyset$; $V' = V$; $i = 1$;\\
 \While{$[V' \neq \emptyset]$}{
 choose a vertex $v \in V'$ of minimum degree in $G[V']$\;
 $d_i = d_{G[V']}(v)$ \;
 $S = S \cup \{v\}$\;
 \eIf{there exists a vertex $p\in S_1$ such that $(p, v) \in E$}{ 
 $S_2 = S_2 \cup \{v, p\}$\; \label{S1toS2}
 $S_1 = S_1 \setminus \{p\}$\;
 $X = X \cup N_{V'}(p) \cup N_{V'}(v)$\; \label{X}
 $V' = V\setminus (S \cup X)$\;
 }{
   $S_1 = S_1 \cup \{v\}$\; 
   $X = X \cup [\bigcup_{t \in S_1}[N_{V'}(v) \cap N_{V'}(t)]]$\;
   $V' = V\setminus (S \cup X)$\;
   }
 $i = i+1$\;
 }
  Return $S$\;
 \caption{\disso{}}
 \label{Disscociation-Algo}
\end{algorithm2e}

The correctness of the algorithm can be observed easily using the above lemma. Consider a solution set $S = S_1 \dot{\cup} S_2$ returned by this algorithm. From Lemma \ref{lemma:S2inX}, we can see that each of the vertices in $S_2$ are in a component of size exactly 2. Now consider a vertex $v \in S_1$. It cannot be a neighbor to a vertex in $S_2$, again by lemma \ref{lemma:S2inX}. In addition, it cannot have a neighbor in $S_1$, since if it did, Line \ref{S1toS2} would have moved these two vertices to the set $S_2$ during the run of the algorithm. Thus we see that the set $S_1$ forms an independent set in $G[S]$ and thus the maximum size of a component in $G[S]$ is 2, implying that $S$ is a feasible solution.
\begin{lemma} \label{lem1}
 $n \leq \sum_{i=1}^q(d_i + 1) \leq 2n$, where $n$ is the number of vertices in the input graph. 
\end{lemma}
\begin{proof}
When we choose a vertex $v \in G[V']$ of minimum degree, we denote it as a {\em chosen} vertex and denote its neighbors in $G[V']$ as {\it looked-at} vertices. From the algorithm it is clear that a {\em chosen} vertex is always included in $S$. Every time a vertex is {\em looked-at}, exactly one of its neighbours is included in $S$. Once a vertex is {\em looked-at} twice, the algorithm includes it in the set $X$. Therefore, a vertex can be {\em looked-at} at most twice and can be {\em chosen} at most once. 

Based on these notions, it follows that $X$ is the set of vertices which are marked as {\em looked-at} twice by the algorithm. In the $i$th step of the algorithm a minimum degree vertex can be a {\em looked-at} vertex in $V'$. 
A chosen vertex $v_i$ can make at most $d_i$ vertices as {\em looked-at} twice. Thus inclusion of $v_i$ into $S$ can add at most $d_i$ vertices into the set $X$. Also the set $X$ is constructed in $q$ steps.
Therefore, $\sum_{i=1}^qd_i \geq |X|$. From this observation it follows that $\sum_{i=1}^q(d_i + 1) \geq |X| + |S| = n.$

The greedy algorithm finally partitions the vertex set $V$ as $S \dot{\cup} X$. Each vertex in $V$ is probed at most twice (as {\em chosen}, or {\em looked-at} and {\em chosen}, or {\em looked-at} and {\em looked-at}) before it is put in $S$ or $X$. At the $i$th step $(d_i +1)$ represents the number of {\em looked-at} vertices plus the {\em chosen} vertex. Therefore, $\sum_{i=1}^{q} (d_i + 1) \leq 2n$. 
\end{proof}

\begin{lemma} \label{lem2}
 $\sum_{i=1}^{q} (d_i + 1)d_i \leq 3n\overline{d}.$
\end{lemma}
\begin{proof}
Every time we choose a vertex $v \in V'=[V\setminus (S\cup X)]$, we assign a token to each edge $e$ incident on a vertex of $N_{V'}[v]$. 

Since the degree of vertex $v$ in the $i$th step is $d_i$ and it is of minimum degree, we assign at least $\frac 12 {d_i(d_i+1)}$ tokens in $i$th step.  We claim that at the end of the algorithm, an edge can have at most 3 tokens. Let $v_i$ be the vertex chosen in $i$th step. Let $(p, q)$ be an edge in $G[V']$ such that $p \in N_{V'}(v_i)$ and $q \notin N_{V'}(v_i)$. While including $v_i$ into $S$, we are assigning a token to $(p, q)$ and it remain in the subgraph $G[V']$ with updated vertex set $V'$. If a neighbor of $p$ or $p$ is chosen as a vertex of smallest degree in any of the subsequent steps then $(p, q)$ will get one more token and the updated graph $G[V']$ will not have this edge. So we assume that in a subsequent $j$th step let $v_j$ be chosen as smallest degree vertex and it a neighbor of $q$ and not a neighbor of $p$. While including $v_j$ in $S$ we are assigning the 2nd token to $(p, q)$ and it remains as an edge in the updated subgraph $G[V']$. Till now both $p$ and $q$ are {\em looked-at} once each. In a subsequent $k$th step when a vertex $u \in N_{V'}[p] \cup N_{V'}[q]$ is chosen and included in $S$, the edge $(p, q)$ receives the 3rd token and it will no more in the updated subgraph $G[V']$ (as $p$ or $q$ will be {\em looked-at} twice and will not be present in the updated set $V'$).

Therefore, we have $\sum_{i=1}^{q} \frac 12 (d_i + 1)d_i \leq 3|E| = \frac{3}{2} n \overline{d}$, and thus $\sum_{i=1}^{q} (d_i + 1)d_i \leq 3n \overline{d}.$
\end{proof}

\begin{lemma} \label{lem3}
The solution set returned by Algorithm \ref{Disscociation-Algo} has size at least $\frac n{3\overline{d}+2}$.
\end{lemma}
\begin{proof}
 By using the inequalities in Lemma \ref{lem1} and Lemma \ref{lem2}, we have
 \begin{align*}
  \sum_{i=1}^q(d_i + 1)^2 &= \sum_{i=1}^qd_i(d_i +1) + \sum_{i=1}^q(d_i + 1) \\
  &\leq 3n \overline{d} + 2n.
 \end{align*}
  From the Cauchy-Schwarz inequality, we get
  \begin{align*}
   \sum_{i=1}^q(d_i + 1)^2 &\geq \frac{\left(\sum_{i=1}^q(d_i +1)\right)^2}{\sum_{i=1}^q1^2}\\
   &\geq \frac{n^2}{q}.
  \end{align*}
  Combining these two inequalities yields $q \geq \frac{n}{3\overline{d} + 2}$.
 \end{proof}

From the above lemma, we have the following result.
\begin{theorem}
\disso{} can be approximated within a factor of $(3 \overline{d} + 2)$, where $\overline{d}$ is the average degree of the input graph $G$.
\end{theorem}

\subsection{Approximation algorithm for \maxkcomp{k}}
In this section, we will generalize the $(3\overline{d} + 2)$ factor approximation algorithm for \disso{} to get an algorithm for \maxkcomp{k} which is a $(2k - 1)\overline{d} + k$ factor approximation algorithm. The following is a generalization of Algorithm \ref{Disscociation-Algo}.

\begin{algorithm2e}[H] 
\SetAlgoLined
\KwIn{An undirected simple graph $G=(V, E)$}
\KwOut{$S \subseteq V$ such that each component in $G[S]$ has at most $k$ vertices}
 $S= \emptyset$; $X=\emptyset$; $V' = V$; $i=1;$\\
 \While{$[V' \neq \emptyset]$}{
 choose a vertex $v \in V'$ of minimum degree in $G[V']$\;
 $d_i = d_{G[V']}(v)$\;
 $S = S \cup \{v\}$\;
 \For{each vertex $p \in V'$ and $p \neq v$}{ \label{algostep:BCISkrem}
   \If{$G[S\cup \{p\}]$ has a component of size at least $k+1$}{
    $X = X \cup\{p\}$\;
   }
 }
 $V' = V\setminus (S \cup X)$\;
 $i= i+1$\;
 }
 Return $S$\;
 \caption{Algorithm for \maxkcomp{k}}
 \label{bcis-Algo}
\end{algorithm2e}

By using a similar argument we have the following results.
\begin{lemma} \label{lem4}
Let $q$ be the number of iterations made by the while loop in the Algorithm \ref{bcis-Algo}. 
\\ 
 (a) $n \leq \sum_{i=1}^q(d_i + 1) \leq nk$ \\
 (b) $\sum_{i=1}^q d_i(d_i + 1) \leq \overline{d} (2k-1)n.$
\end{lemma}

\begin{proof}
(a) By assigning token to the vertices, as given in the proof of Lemma \ref{lem1}, it can be proved that $\sum_{i=1}^q d_i \geq |X|$. Hence, we get $n \leq \sum_{i=1}^q(d_i + 1)$. For the other inequality, consider the following argument. In each iteration, when a vertex $v$ is picked, assign one token to that vertex and each of its neighbors in the graph $G[V']$. Thus, in the $i$th iteration, the number of tokens assigned is exactly $d_i + 1$, and the total number of tokens assigned is $\sum_{i=1}^q (d_i + 1)$. Now, we look at the maximum number of tokens that any particular vertex can get. Note that a vertex gets a token if it is either {\em chosen} or {\em looked-at}, and once a vertex is {\em chosen} or included in the set $X$, it gets no more tokens. Suppose that at the end of the algorithm, a vertex $u$ has $(k+1)$ or more tokens. This implies that either it was {\em looked-at} at least $k+1$ times and then moved to the set $X$ or it was {\em looked-at} at least $k$ times and then chosen. Neither of these cases is possible, since when the vertex $u$ is {\em looked-at} $k$ times $k$ of its neighbors have been {\em chosen}, and thus the Line \ref{algostep:BCISkrem} will remove $u$ after $k$ of its neighbors have been {\em chosen}. Thus, any vertex can have at most $k$ tokens, and thus the total number of tokens available is at most $nk$. This proves the second inequality. 

(b) In the $i$th iteration, assume that vertex $v$ is {\em chosen}. Assign a token to all the edges incident on $v$ and its neighbors from the set $V'$ of the $i$th iteration. As argued in the case of the dissociation set problem, due to minimality of $d_i$, we assign at least $\sum_{i=1}^q \frac 12 d_i(d_i +1)$ tokens. Now, we produce an upper bound on the total number of tokens by seeing the maximum number of tokens that can be assigned to any edge. Consider an edge $e = (u,v)$ - it gets a token whenever at least one of its endpoints is either {\em looked-at} or {\em chosen}. Now, from the proof of (a), it follows that any vertex is {\em looked-at} at most $k$ times and put into the set $X$,  or {\em looked-at} at most $k-1$ times and then {\em chosen}. Thus the edge $e$ can have at most $(2k-1)$ tokens, since if it receives that amount, then at least one of its endpoints has been moved to either the set $X$ or the set $S$ and the edge no longer receives any tokens. Thus the total number of tokens that any edge can receive is $(2k-1)$ and the total number of tokens given out is at most $|E|(2k-1)$. Thus we get $\sum_{i=1}^q d_i(d_i + 1) \leq 2|E|(2k-1) = \overline{d} (2k -1) n$. 
\end{proof}

By using the inequalities in Lemma \ref{lem4} and Cauchy-Schwarz inequality, we have
\begin{theorem}
 \maxkcomp{k} can be approximated within a factor of $(2k - 1)\overline{d} + k$, where $\overline{d}$ is the average degree of the input graph $G$.
\end{theorem}

\section{Upper bound for \maxkcomp{k} with large $\Delta$}

In this section, we provide a reduction from the \disso{} problem to the \textsc{Max-IS} problem, and use the upper bound result for \textsc{Max-IS} from~\cite{halldorsson2004approximations} to provide a $O(\frac{\Delta \log \log \Delta}{\log \Delta})$ approximation factor, for sufficiently large values of $\Delta$. 

\begin{lemma} \label{lem:BCIS2-IS-redn}
 \disso{} $\leq_{\text{\sf AP}}$ \textsc{Max-IS}.
\end{lemma}
\begin{proof}
The reduction is a straight forward one, where we solve the \textsc{Max-IS} on the input instance. Formally, the reduction is the identity mapping, and given an instance $G=(V, E)$ of \disso{}, it maps $G=(V,E)$ as an instance of \textsc{Max-IS}.

Given an independent set $I$ in $G$, we construct a dissociation set $S = I$ from $I$ with $|S|=|I|$. Given a maximum cardinality dissociation set $S^*$ in $G$, we construct an independent set $I$ with $|I| \geq \frac12|S^*|$ as follows.
The induced graph $G[S^*]$ will have maximum degree of at most 1, and we can partition the set $S^*$ into two sets $S_0 = \{v \in S^* | d_{G[S^*]}(v) = 0\}$ and $S_1 = \{v \in S^* | d_{G[S^*]}(v) = 1\}$. It is important to observe that $G[S_1]$ forms an induced matching. We construct an independent set $I$ in $G$ by taking all the vertices in $S_0$ and exactly one vertex from each edge in $G[S_1]$. This set $I$ of vertices is an independent set and $|I| = |S_0| + \frac12|S_1| \geq \frac{1}{2}(|S_0| + |S_1|) = \frac{1}{2}|S^*|$. Thus we have $|S^*| \leq 2|I| \leq 2\alpha(G)$, yielding 
\[\frac{|S^*|}{|S|} \leq 2\frac{\alpha(G)}{|I|}.\]
\end{proof}

\begin{theorem} {\rm~\cite{halldorsson2004approximations}} 
\label{thm:IS-largeDelta-upperbound}
{\sc Max-IS} can be approximated within a factor of $O(\frac{\Delta \log \log \Delta}{\log \Delta})$ for sufficiently large values of $\Delta$, where $\Delta$ is the maximum degree of the input graph.
\end{theorem}

From Lemma \ref{lem:BCIS2-IS-redn} and Theorem \ref{thm:IS-largeDelta-upperbound}, we get the following result,
\begin{theorem} \label{thm:BCIS2-largeDelta-UB}
\disso{} is approximable within a factor of $O(\frac{\Delta \log \log \Delta}{\log \Delta})$ for sufficiently large values of $\Delta$, where $\Delta$ is the maximum degree of the input graph.
\end{theorem}

The reduction in the proof of Lemma \ref{lem:BCIS2-IS-redn} is also a reduction from \maxkcomp{k} to \textsc{Max-IS}. Here, given an instance $G=(V, E)$ of \maxkcomp{k}, we take $G=(V, E)$ as an instance of \textsc{Max-IS}. If $I$ is an independent set in $G$ then $S = I$ is also a $k$-component set in $G$. From a maximum cardinality $k$-component set  $S^*$ in the graph $G$, we can construct an independent set $I$ in $G$ with $|I| \geq \frac{1}{k}|S^*|$. Therefore, we have  $|S^*| \leq k|I| \leq k\alpha(G)$. From these inequalities we have 
\[\frac{|S^*|}{|S|} \leq k\frac{\alpha(G)}{|I|}\] and the following result.

\begin{theorem} \label{thm:BCISk-largeDelta-UB}
 For $k>0$, \maxkcomp{k} is approximable within a factor of $O(\frac{k\Delta \log \log \Delta}{\log \Delta})$ for sufficiently large values of $\Delta$, where $\Delta$ is the maximum degree of the input graph.
\end{theorem}

\section{Conclusion}

We have proved that \maxkcomp{k} is hard to approximate like \textsc{Max-IS}. We also show that for $k= O(\frac{n}{\log n})$, \maxkcomp{k} is hard to approximate within a factor of $O(\log n)$. We generalize the greedy algorithm for \textsc{Max-IS} to obtain an algorithm for \maxkcomp{k} with an approximation factor of $(2k -1)\overline{d} + k$. This approximation factor is a generalization of Tur\'{a}n's bound for \textsc{Max-IS}.

\end{document}